\documentclass[12pt]{article}
\usepackage[active]{srcltx}
\usepackage{amssymb,latexsym,amsmath,enumerate,verbatim,amsfonts,amsthm}
\usepackage{cite}
\usepackage[ansinew]{inputenc}
\usepackage{graphicx}
\usepackage{color}
\usepackage{url}
\usepackage[colorlinks]{hyperref}
\numberwithin{equation}{section}
\newtheorem{theorem}{Theorem}[section]

\newtheorem{lemma}{Lemma}[section]

\newtheorem{remark}{Remark}[section]
\begin{document}
\title{Copositivity for 3rd order symmetric tensors and applications}%
	
\author{Jiarui Liu, Yisheng Song\thanks{School of Mathematics and Information Science  and Henan Engineering Laboratory for Big Data Statistical Analysis and Optimal Control, Henan Normal University, XinXiang HeNan,  P.R. China, 453007.
			Email: songyisheng@htu.cn. This author's work was supported by
			the National Natural Science Foundation of P.R. China (Grant No.
			11571095, 11601134).}}
	\date{\today}
	
	\maketitle
	%
	\vskip 4mm
\begin{quote}{\bf Abstract.}\ The strict copositivity of 4th order symmetric tensor may apply to detect vacuum stability of general scalar potential. For finding analytical expressions of (strict) copositivity of 4th order symmetric tensor, we may reduce its order to 3rd order to better deal with it.  So,  it is provided that several analytically sufficient conditions for the copositivity of 3rd order 2 dimensional (3 dimensional) symmetric tensors. Subsequently, applying these conclusions to 4th order  tensors, the analytically sufficient conditions of copositivity are proved for 4th order 2 dimensional and 3 dimensional symmetric tensors.  Finally, we apply these results to present analytical vacuum stability conditions for vacuum stability for $\mathbb{Z}_3$ scalar dark matter.
	\vskip 2mm
		
	\vskip 2mm
		
\end{quote}
	\vskip 8mm

	\section{\bf Introduction}\label{}
	
	Recently, Kannike \cite{K2016,K2018,K2012} studied the vacuum stability of general scalar potentials
	of a few fields.  The most general scalar quartic
potential of the {\bf SM} Higgs $\mathbf{H}_1$, an inert doublet $\mathbf{H}_2$ and a
complex singlet $\mathbf{S}$ is
\begin{align}
V(h_1,h_2,s)=\mathbf{\Gamma}\phi^4=&\sum_{i,j,k,l=1}^3 \gamma_{ijkl}z_iz_jz_kz_l,\label{eq:11}
\end{align}
where $z=(z_1,z_2,z_3)^\top=(h_1,h_2,s)^\top$, $
h_1=|H_1|,\ h_2=|H_2|, H_2^{\dagger} H_1=h_1h_2\rho e^{i\phi}, S=se^{i\phi_S},$
 $\mathbf{\Gamma}=(\gamma_{ijkl})$ is coupling tensor, an 4th order 3 dimensional real symmetric tensor.  Clearly, $h_1\geq0, h_2\geq0, s\geq0.$   So, the vacuum stability of $\mathbb{Z}_3$ scalar dark matter $V(h_1,h_2,s)$ is equivalent to the strict copositivity of the coupling tensor $\mathbf{\Gamma}=(\gamma_{ijkl})$.  The concept of (strict) copositivity of symmetric tensors was  introduced by Qi \cite{LQ5} in 2013.  A symmetric tensor $\mathbf{\Gamma}$ with order $m$ and dimension $n$ is called \begin{itemize}
		\item[(i)] {\em copositive } if $\mathbf{\Gamma}x^m=\sum\limits_{i_1,i_2,\cdots,i_m=1}^n\gamma_{i_1i_2\cdots
		i_m}x_{i_1}x_{i_2}\cdots
	x_{i_m}\geq0$ for all nonnegative vector $x=(x_1,\cdots,x_n)^\top$;
\item[(ii)] {\em strictly copositive} if  $\mathbf{\Gamma}x^m>0$ for all nonnegative and nonzero vector $x=(x_1,\cdots,x_n)^\top$.\end{itemize}

Since the practical matters such as the vacuum stability of general scalar potentials of a few fields require precise expressions, it is necessery to find the analytical expressions of the strict copositivity of a symmetric tensor.
 Qi \cite{LQ5} showed a symmetric tensor is strictly copositive if for all $i=1,2,\cdots, n$, $$\gamma_{ii\cdots i}+\sum\{\gamma_{i_1i_2\cdots
		i_m}<0,\ (i_1,i_2,\cdots, i_m)\ne(i,i,\cdots, i)\}>0.$$ Song-Qi \cite{SQ2015} gave a necessary and sufficient sondition of strictly copositive tensors with the help of its H-eigenvalue and Z-eigenvalue. Song-Qi \cite{S-Q2016} studied the (strict) copositvity of symmetric tensors by means of the constrained minimization problem on the unit sphere. Song-Qi \cite{SQ2016} showed that the (strict) copositivity of a symmetric tensor is equivalent to its (strict) semipositiveness. A tensor $\mathbf{\Gamma}$ is said to be (strictly) semipositive (Song-Qi \cite{SQ2017}) if for each nonzero and nonnegative vector $x=(x_1,x_2,\cdots,x_n)^\top$, there is $k\in \{1,2,\cdots,n\}$ such that $$x_k>0\mbox{ and }\left(\mathbf{\Gamma} x^{m-1}\right)_k=\sum_{i_2,\cdots,i_m=1}^n\gamma_{ki_2\cdots i_m}x_{i_2}\cdots x_{i_m}\geq0(>0).$$ This class of tensors has close relationships with  the tensor complementarity problems (for short, TCP)(\cite{SY2016,S-Q2017,SM2018}).  For  more details about TCP, also see \cite{BHW2016,BP2018,CQS2018,CQW2016,CQW2019,DLQ2018,G2017,HQ2017,LQX2017,WCW2018,WHB2016,WHH2018} and references cited therein.

Recently, on the base of the main results of Song-Qi \cite{SQ2016}, an analytical expression of detecting the strict copositivity of a symmetric tensor was  obtained by summarizing conclusions of Qi-Song \cite{QS}, Song-Qi \cite{S-Q2015}, Song-Mei \cite{SM2018}, Yuan-You\cite{YY}. That is,  a symmetric tensor $\mathbf{\Gamma}$ is strictly copositive if it satisfies \begin{itemize}
\item[(i)]  $\sum\limits_{i_{2},\cdots,i_{m}=1}^{n}\gamma_{ii_{2}i_{3}\cdots i_{m}}>0$ for all $i\in \{1,2,\cdots,n\},$
\item[(ii)] $\frac{1}{n^{m-1}}(\sum\limits_{i_{2},\cdots,i_{m}=1}^{n}\gamma_{ii_{2}i_{3}\cdots i_{m}})>\gamma_{ij_{2}j_{3}\cdots j_{m}}$  for all $(j_{2},j_{3},\cdots,j_{m})\neq (i,i,\cdots,i)$.\end{itemize}
The numerical algorithms of copositivity of a tensor also were presented by  Chen-Huang-Qi \cite{CHQ2017}, Chen-Huang-Qi \cite{CHQ2018}, Nie-Yang-Zhang \cite{NYZ2018}, Li-Zhang-Huang-Qi \cite{LZHQ2019}, also see \cite{CW2018, QCC2018, QL2017}.
  Very recently, Song-Qi \cite{SQ2019} provided several analytically sufficient conditions of checking (strict) copositivity of 4th order 3 dimensional symmetric tensors by reducing orders or dimensions  of tensor.

	In this paper, we will study the analytical expressions of certifying copositivity of a 3rd order 3 (or 2) dimensional symmetric tensor. Then applying them to show sufficient conditions of copositivity of 4th order 3 dimensional tensors, which are differ from ones of Song-Qi \cite{SQ2019}. Furthermore, we use these conclusions to check  vacuum stability for $\mathbb{Z}_3$ scalar dark matter.
	
	\section{\bf Preliminaries and Basic facts}

Let $\|\cdot\|$ denote any norm on $\mathbb{R}^n$. Then the equivalent definition of (strict) copositivity and semipositive (positive) definiteness of a symmetric tensor was presented in the sense of any norm on $\mathbb{R}^2$ \cite{LQ1,LQ5,SQ2015, QL2017,QCC2018}.

	\begin{lemma}(\cite{SQ2015}) \label{le:21} Let $\mathbf{\Gamma}$ be a symmetric tensor of order $m$ and dimension $n$. Then
		\begin{itemize}
			\item[(i)] $\mathbf{\Gamma}$ is copositive if and only if  $\mathbf{\Gamma}x^m\geq0$ for all nonnegative vectors $x\in \mathbb{R}^n$ with $\|x\|=1$;
			\item[(ii)] $\mathbf{\Gamma}$ is strictly copositive if and only if $\mathbf{\Gamma}x^m>0$ for all nonnegative vectors $x\in \mathbb{R}^n$ with $\|x\|=1$.
		\end{itemize}
	\end{lemma}

For a cubic and univariate polynomial $P(t)$ with real coefficients,
\begin{equation}\label{eq:25}P(t)=at^3+bt^2+ct+d,\end{equation}
Schmidt-He$\beta$ \cite{SH1988} proved its nonnegative conditions.

\begin{lemma}(\cite[Proposition 2]{SH1988}) \label{le:22} Let $P(t)$ be a cubic and univariate polynomial given by \eqref{eq:25}. Then (i) $P(t)\geq0$ for all $t\geq 0$ if and only if  the inequalities systems (1) or (2) hold, \begin{itemize}
	\item[(1)]\  $ a\geq0,\  b\geq0,\ c\geq0,\ d\geq0,$
	\item[(2)]\ $\max\{a,d\}>0$,\ $ a\geq0,\  d\geq0,\ 4ac^3+4b^3d+27a^2d^2-18abcd-b^2c^2\geq0.$
\end{itemize}
(ii) $P(t)\geq0$  for all $t\geq 0$ if the inequalities
$$a\geq0,\  d\geq0,\ b\geq a-2\sqrt{ad}, c\geq d-2\sqrt{ad}$$
hold simultaneously.
\end{lemma}

\begin{lemma}(\cite[Proposition 3]{SH1988}) \label{le:23} Let a quadratic  and univariate  polynomial $p(t)$ be given by
\begin{equation}\label{eq:23}p(t)=\alpha t^2+\beta t+\gamma.\end{equation} Then $p(t)\geq0$  for all $t\geq0$ if and only if the inequalities
\begin{equation}\label{eq:24}\alpha\geq0,\ \gamma\geq0,\ \beta+2\sqrt{\alpha\gamma}\geq0\end{equation}
hold simultaneously.
\end{lemma}

\section{\bf Copositivity of 3rd order tensors}

Let $\mathbf{\Gamma}=(\gamma_{ijk})$ be a 3rd order $2$ dimensional symmetric  tensor. Then for a vector $x=(x_1,x_2)^\top$,
\begin{equation}\label{eq:31}\begin{aligned}\mathbf{\Gamma}x^3&=\sum_{i,j,k=1}^2\gamma_{ijk}x_ix_jx_k\\
&=\gamma_{111}x_1^3+3\gamma_{112}x_1^2x_2+3\gamma_{122}x_1x_2^2+\gamma_{222}x_2^3.\end{aligned}\end{equation}

\begin{theorem} \label{th:31} A 3rd order $2$ dimensional symmetric tensor is copositive if and only if the following inequalities systems (1) or (2) hold,
	\begin{itemize}
	\item[(1)]\ $\gamma_{111}\geq0$,\ $\gamma_{222}\geq0$, \ $\gamma_{112}\geq0,\ \gamma_{122}\geq0;$
	\item[(2)]\ $\max\{\gamma_{111}, \gamma_{222}\}>0$,  $\gamma_{111}\geq0$,\ $\gamma_{222}\geq0$,\\ $4\gamma_{111}\gamma_{122}^3+4\gamma_{112}^3\gamma_{222}+\gamma_{111}^2\gamma_{222}^2-6\gamma_{111}\gamma_{112}\gamma_{122}\gamma_{222}-3\gamma_{112}^2\gamma_{122}^2 \geq0.$
		\end{itemize}
\end{theorem}
	
	\begin{proof}  It follows from Lemma \ref{le:21} that we can restrict $x=(x_1,x_2)^\top$ to $$\|x\|=|x_1|+|x_2|=1
\mbox{ for all } x\mbox{ with }x_1\geq0,\ x_2\geq0.$$
Clearly, $\mathbf{\Gamma}x^3=\gamma_{111}x_1\geq0$ for $x=(x_1,0)$ and $\mathbf{\Gamma}x^3=\gamma_{222}x_2\geq0$ for $x=(0,x_2)$. Now  both $x_1$ and $x_2$ are not $0$.
 Then the homogeneous polynomial $\mathcal{A}x^3$ can be divided by $x_2^3$ to  yield $$\frac{\mathbf{\Gamma}x^3}{x_2^3}=\gamma_{111}\left(\frac{x_1}{x_2}\right)^3+3\gamma_{112}\left(\frac{x_1}{x_2}\right)^2+3\gamma_{122}\left(\frac{x_1}{x_2}\right)+\gamma_{222}.$$
	Let $P(t)=\frac{\mathbf{\Gamma}x^3}{x_2^3}$, i.e.,
\begin{equation}\label{eq:32}P(t)=\gamma_{111}t^3+3\gamma_{112}t^2+3\gamma_{122}t+\gamma_{222}.\end{equation}
	Clearly, $P(t)\geq0$ if and only if $\mathbf{\Gamma}x^3\geq0$. An application of Lemma \ref{le:22}(i) to the polynomial $P(t)$ with $(a,b,c,d)^\top=(\gamma_{111},3\gamma_{112},3\gamma_{122},\gamma_{222})^\top$, we obtain  that $P(t)\geq0$ for all $t\geq0$ if and only if the inequalities systems $$a=\gamma_{111}\geq0, \ d=\gamma_{222}\geq0,\ b=3\gamma_{112}\geq0,\ c=3\gamma_{221}\geq0$$ or
$$\begin{aligned}&\max\{\gamma_{111}, \gamma_{222}\}>0,\ \gamma_{111}\geq0,\ \gamma_{222}\geq0,\\ &4\gamma_{111}(3\gamma_{122})^3+4(3\gamma_{112})^3\gamma_{222}+27\gamma_{111}^2\gamma_{222}^2-18\gamma_{111}(3\gamma_{112})(3\gamma_{122})\gamma_{222}\\
&-(3\gamma_{112})^2(3\gamma_{122})^2 \geq0.
\end{aligned}$$
Namely, \begin{itemize}
	\item[ ]\ \ \ $\gamma_{111}\geq0$,\   $\gamma_{222}\geq0$, \ $\gamma_{112}\geq0,\ \gamma_{122}\geq0;$
	\item[or]\ \ $\max\{\gamma_{111}, \gamma_{222}\}>0$, \ $\gamma_{111}\geq0$,\ $\gamma_{222}\geq0$,\ $$4\gamma_{111}\gamma_{122}^3+4\gamma_{112}^3\gamma_{222}+\gamma_{111}^2\gamma_{222}^2-6\gamma_{111}\gamma_{112}\gamma_{122}\gamma_{222}-3\gamma_{112}^2\gamma_{122}^2 \geq0.$$
\end{itemize}
Therefore, the desired conclusions are obtained.
		\end{proof}

\begin{remark}	
It follows from the proof of Theorem \ref{th:31} that $\mathbf{\Gamma}x^3$ may be divided by $x_1^3$ ($x_1\ne0$), then $$P(t)=\frac{\mathbf{\Gamma}x^3}{x_1^3}=\gamma_{111}+3\gamma_{112}t+3\gamma_{122}t^2+\gamma_{222}t^3,$$ where $t=\frac{x_2}{x_1}$. The conclusions are same.	
\end{remark}
	
\begin{theorem} \label{th:32} Let $\mathbf{\Gamma}$ be a 3rd order 2 dimensional and symmetric tensor. Assume that
		$$\begin{aligned}
		&\gamma_{111}\geq0, \gamma_{222}\geq0,\\
		&\gamma_{112}\geq \frac{\gamma_{111}-2\sqrt{\gamma_{111}\gamma_{222}}}3,\ \gamma_{122}\geq \frac{\gamma_{222}-2\sqrt{\gamma_{111}\gamma_{222}}}3.
\end{aligned}$$
		Then  $\mathbf{\Gamma}$ is copositive.
	\end{theorem}
	
	\begin{proof}  Using the proof technique of Theorem \ref{th:31}, we only need to show the nonnegativity of the polynomial $P(t)$ given by \eqref{eq:32} for all $t\geq0,$ where $$P(t)=at^3+bt^2+ct+d,\  a=\gamma_{111},\ b=3\gamma_{112},\ c=3\gamma_{122}, \ d=\gamma_{222}$$ The assumptions mean that $a=\gamma_{111}\geq0$, $d=\gamma_{222}\geq0,$
$$b=3\gamma_{112}\geq  \gamma_{111}-2\sqrt{\gamma_{111}\gamma_{222}},\ c=3\gamma_{122}\geq  \gamma_{222}-2\sqrt{\gamma_{111}\gamma_{222}}.$$
From Lemma \ref{le:22}(ii), it follows that $P(t)\geq0$ for all $t\geq0,$  and so, the tensor $\mathbf{\Gamma}$ is copositive, as required.
	\end{proof}

	\begin{theorem} \label{th:33} Let $\mathbf{\Gamma}$ be a 3rd order 2 dimensional and symmetric tensor. Assume that one of the inequalities systems (1) and (2) holds,
		\begin{itemize}
		  \item[(1)]\ \  $ \gamma_{111}\geq0,\  \gamma_{222}\geq0,\  \gamma_{122}\geq0,\ \gamma_{112}\geq-\frac23\sqrt{3\gamma_{122}\gamma_{111}}$;
		  \item[(2)]\ \  $ \gamma_{111}\geq0,\  \gamma_{222}\geq0,\  \gamma_{112}\geq0,\ \gamma_{122}\geq -\frac23\sqrt{3\gamma_{112}\gamma_{222}}$.
		\end{itemize}
		Then  $\mathbf{\Gamma}$ is  copositive.
		\end{theorem}

\begin{proof}  Using the proof technique of Theorem \ref{th:31}, it is known that the copositivity of $\mathbf{\Gamma}$ is equivalent to the nonnegativity of the polynomial $P(t)$ for all $t\geq0,$  where $$P(t)=\gamma_{111}t^3+3\gamma_{112}t^2+3\gamma_{122}t+\gamma_{222}.$$
Now we show the nonnegativity of $P(t)$. In fact, \begin{align}
P(t)=&\ \gamma_{111}t^3+(3\gamma_{112}t^2+3\gamma_{122}t+\gamma_{222}),\nonumber\\
    =&\ t(\gamma_{111}t^2+3\gamma_{112}t+3\gamma_{122})+\gamma_{222}.\nonumber
\end{align}
Let $f(t)=3\gamma_{112}t^2+3\gamma_{122}t+\gamma_{222}$ and $g(t)=\gamma_{111}t^2+3\gamma_{112}t+3\gamma_{122}.$
Then from Lemma \ref{le:23}, it follows that the nonnegativity of $f(t)$ is equivalent to the inequalities system,
\begin{equation}\label{eq:33}\gamma_{222}\geq0,\  \gamma_{112}\geq0,\ 3\gamma_{122}+2\sqrt{3\gamma_{112}\gamma_{222}}\geq0.\end{equation}
So, the inequality $\gamma_{111}\geq0$ along with the above inequalities system \eqref{eq:33} imply the polynomial $P(t)=\gamma_{111}t^3+ f(t)\geq0$ for all $t\geq0,$ that is, $\mathbf{\Gamma}$ is copositive.

For the aussumption (1), the same technique is applied to $g(t)$ to yield the desired conclusion.
	\end{proof}

For an 3rd order $3$ dimensional symmetric  tensor $\mathbf{\Gamma}$ and a vector $x=(x_1, x_2, x_3)^\top$,
\begin{equation}\label{eq:34}\begin{aligned}\mathbf{\Gamma}x^3=&\sum_{i,j,k=1}^3\gamma_{ijk}x_ix_jx_k\\
=&\gamma_{111}x_1^3+\gamma_{222}x_2^3+\gamma_{333}x_3^3+3\gamma_{112}x_1^2x_2+3\gamma_{122}x_1x_2^2+3\gamma_{113}x_1^2x_3\\
&+3\gamma_{133}x_1x_3^2+3\gamma_{223}x_2^2x_3+3\gamma_{233}x_2x_3^3+6\gamma_{123}x_1x_2x_3.\end{aligned}\end{equation}

\begin{theorem} \label{th:34} An 3rd order $3$ dimensional symmetric  tensor $\mathbf{\Gamma}$ is copositive if the following inequalities systems hold,
$$\begin{aligned}
	\gamma_{111}\geq0,\ \gamma_{222}\geq0, \ \gamma_{333}\geq0,\ \gamma_{123}\geq0,\\
	32\gamma_{111}\gamma_{122}^3+32\gamma_{112}^3\gamma_{222}+\gamma_{111}^2\gamma_{222}^2-24\gamma_{111}\gamma_{112}\gamma_{122}\gamma_{222}
 -48\gamma_{112}^2\gamma_{122}^2 \geq0,\\
 32\gamma_{111}\gamma_{133}^3+32\gamma_{113}^3\gamma_{333}+\gamma_{111}^2\gamma_{333}^2-24\gamma_{111}\gamma_{113}\gamma_{133}\gamma_{333}
 -48\gamma_{113}^2\gamma_{133}^2 \geq0,\\
 32\gamma_{222}\gamma_{233}^3+32\gamma_{223}^3\gamma_{333}+\gamma_{222}^2\gamma_{333}^2-24\gamma_{222}\gamma_{223}\gamma_{233}\gamma_{333}
 -48\gamma_{223}^2\gamma_{233}^2 \geq0.
		\end{aligned}$$
\end{theorem}
	
	\begin{proof} Rewritten $\mathbf{\Gamma}x^3$ as follows
$$\begin{aligned}\mathbf{\Gamma}x^3=&(\frac12\gamma_{111}x_1^3+3\gamma_{112}x_1^2x_2+3\gamma_{122}x_1x_2^2+\frac12\gamma_{222}x_2^3)\\
+&(\frac12\gamma_{111}x_1^3+3\gamma_{113}x_1^2x_3+3\gamma_{133}x_1x_3^2+\frac12\gamma_{333}x_3^3)\\
+&(\frac12\gamma_{222}x_2^3+3\gamma_{223}x_2^2x_3+3\gamma_{233}x_2x_3^3+\frac12\gamma_{333}x_3^3)\\&
+6\gamma_{123}x_1x_2x_3\\
=&\mathcal{A}y^3+\mathcal{B}z^3+\mathcal{C}w^3+6\gamma_{123}x_1x_2x_3,
\end{aligned}$$
where $y=(x_1, x_2)^\top$, $z=(x_1, x_3)^\top$, $w=(x_2, x_3)^\top$, $\mathcal{A}=(a_{ijk})$ and $\mathcal{B}=(b_{ijk})$ and $\mathcal{C}=(c_{ijk})$ are three 3rd order $2$ dimensional symmetric  tensors with their entries
$$\begin{aligned}
a_{111}=\frac12\gamma_{111},\ a_{112}=\gamma_{112},\ a_{122}=\gamma_{122},\ a_{222}=\frac12\gamma_{222};\\
b_{111}=\frac12\gamma_{111},\ b_{112}=\gamma_{113},\ b_{122}=\gamma_{133},\ b_{222}=\frac12\gamma_{333};\\
c_{111}=\frac12\gamma_{222},\ c_{112}=\gamma_{223},\ c_{122}=\gamma_{233},\ c_{222}=\frac12\gamma_{333}.\\
\end{aligned}$$
 For the polynomial $\mathcal{A}y^3=\frac12\gamma_{111}x_1^3+3\gamma_{112}x_1^2x_2+3\gamma_{122}x_1x_2^2+\frac12\gamma_{222}x_2^3$,  the assumptions imply that $$\begin{aligned}&a_{111}=\frac12\gamma_{111}\geq0,\ a_{222}=\frac12\gamma_{222}\geq0,\\ &4a_{111}a_{122}^3+4a_{112}^3a_{222}+a_{111}^2a_{222}^2-6a_{111}a_{112}a_{122}a_{222}-3a_{112}^2a_{122}^2\\
 &=\frac1{16}(32\gamma_{111}\gamma_{122}^3+32\gamma_{112}^3\gamma_{222}+\gamma_{111}^2\gamma_{222}^2-24\gamma_{111}\gamma_{112}\gamma_{122}\gamma_{222}\\
 &\ \ \ \ \ \ \ \ \ \ \ -48\gamma_{112}^2\gamma_{122}^2) \geq0.
 \end{aligned}$$
 From Theorem \ref{th:31}, it follows that the tensor $\mathcal{A}$ is copositive, i.e., $\mathcal{A}y^3\geq0$ for all $y\geq0.$

 Similarly, we also have $\mathcal{B}z^3\geq0$ for all $z\geq0$ and $\mathcal{C}w^3\geq0$ for all $w\geq0.$ So,
 $$\mathbf{\Gamma}x^3=\mathcal{A}y^3+\mathcal{B}z^3+\mathcal{C}w^3+6\gamma_{123}x_1x_2x_3\geq0\mbox{ for all }x\geq0.$$The desired conclusions are proved.
		\end{proof}

Using the similar proof technique of Theorem \ref{th:34}, we may apply  Theorem \ref{th:32} to three tensors $\mathcal{A}=(a_{ijk})$ and $\mathcal{B}=(b_{ijk})$ and $\mathcal{C}=(c_{ijk})$, and then the following conclusions are showed easily.

\begin{theorem} \label{th:35} Let $\mathbf{\Gamma}$ be a 3rd order 3 dimensional and symmetric tensor. Assume that
		$$\begin{aligned}
		&\gamma_{111}\geq0, \gamma_{222}\geq0,\ \gamma_{333}\geq0,\ \gamma_{123}\geq0,\\
		&\gamma_{112}\geq \frac{\gamma_{111}-2\sqrt{\gamma_{111}\gamma_{222}}}6,\ \gamma_{122}\geq \frac{\gamma_{222}-2\sqrt{\gamma_{111}\gamma_{222}}}6,\\
&\gamma_{113}\geq \frac{\gamma_{111}-2\sqrt{\gamma_{111}\gamma_{333}}}6,\ \gamma_{133}\geq \frac{\gamma_{333}-2\sqrt{\gamma_{111}\gamma_{333}}}6\\
&\gamma_{223}\geq \frac{\gamma_{222}-2\sqrt{\gamma_{222}\gamma_{333}}}6,\ \gamma_{233}\geq \frac{\gamma_{333}-2\sqrt{\gamma_{222}\gamma_{333}}}6.
\end{aligned}$$
		Then  $\mathbf{\Gamma}$ is copositive.
	\end{theorem}


\section{Copositivity of 4th order tensors}	

Let $\mathcal{A}$ be an 4th order $2$ dimensional symmetric tensor. Then for a vector $x=(x_1,x_2)^\top$,
\begin{equation}\label{eq:41}\begin{aligned}\mathcal{A}x^4&=\sum_{i,j,k,l=1}^2a_{ijkl}x_ix_jx_kx_l\\
&=a_{1111}x_1^4+4a_{1211}x_1^3x_2+6a_{1221}x_1^2x_2^2+4a_{1222}x_1x_2^3+a_{2222}x_2^4.\end{aligned}\end{equation}

	\begin{theorem} \label{th:41} Let $\mathcal{A}$ be an 4th order $2$ dimensional symmetric tensor with $a_{1111}>0$ and $ a_{2222}>0$. Assume that one of the following conditions (1) and (2) holds,
\begin{itemize}
  \item[(1)]  $a_{1222}\geq0,\ 54a_{1111} a_{1221}^3+64a_{1121}^3a_{1222}+27a_{1111}^2a_{1222}^2-108a_{1111}a_{1121}a_{1221}a_{1222}\\ -36a_{1121}^2a_{1221}^2 \geq0;$
\item[(2)]  $a_{1112}\geq0, \  64a_{1211} a_{1222}^3+54a_{1221}^3a_{2222}+27a_{1211}^2a_{2222}^2-108a_{1211}a_{1221}a_{1222}a_{2222}\\-36a_{1221}^2a_{1222}^2 \geq0.$
 \end{itemize}
		Then $\mathcal{A}$ is copositive.
	\end{theorem}
\begin{proof}   Rewritten $\mathcal{A}x^4$ as follows,
$$\begin{aligned}
\mathcal{A}x^4=& x_1(a_{1111}x_1^3+4a_{1211}x_1^2x_2+6a_{1221}x_1x_2^2+4a_{1222}x_2^3)+a_{2222}x_2^4\\
    = & a_{1111}x_1^4+x_2(4a_{1211}x_1^3+6a_{1221}x_1^2x_2+4a_{1222}x_1x_2^{2}+a_{2222}x_2^3).
\end{aligned}$$
		Let $$f(x_1,x_2)=a_{1111}x_1^3+4a_{1211}x_1^2x_2+6a_{1221}x_1x_2^2+4a_{1222}x_2^3.$$ Then $f(x_1,x_2)$ can be a homogeneous polynomial defined by a 3rd order 2 dimensional and symmetric tensor $\Gamma=(\gamma_{ijk})$ with its entries
$$\gamma_{111}=a_{1111},\ \gamma_{112}=\frac43a_{1211},\ \gamma_{122}=2a_{1221},\ \gamma_{222}=4a_{1222}.$$
From the assumption (1), it follows that
\begin{align*}
  &4\gamma_{111}\gamma_{122}^3+4\gamma_{112}^3\gamma_{222}+\gamma_{111}^2\gamma_{222}^2-6\gamma_{111}\gamma_{112}\gamma_{122}\gamma_{222}-3\gamma_{112}^2\gamma_{122}^2\\
   &= 4a_{1111}(2a_{1221})^3+4(\frac43a_{1211})^34a_{1222}+a_{1111}^2(4a_{1222})^2\\
   &\ \ \ -6a_{1111}(\frac43a_{1211})(2a_{1221})(4a_{1222})-3(\frac43a_{1211})^2(2a_{1221})^2\\
   &=\frac{16}{9}( 54a_{1111} a_{1221}^3+64a_{1121}^3a_{1222}+27a_{1111}^2a_{1222}^2\\
            &\ \ \ -108a_{1111}a_{1121}a_{1221}a_{1222}-36a_{1121}^2a_{1221}^2)\\& \geq0.
\end{align*}
By Theorem \ref{th:31}, we have $f(x_1,x_2)\geq0$ for all $x\geq0$, and hence, $$\mathcal{A}x^4=x_1f(x_1,x_2)+a_{2222}x_2^4\geq0\mbox{ for all }x=(x_1,x_2)^\top\geq0.$$
That is, $\mathcal{A}$ is copositive.

Let $$g(x_1,x_2)=4a_{1211}x_1^3+6a_{1221}x_1^2x_2+4a_{1222}x_1x_2^{2}+a_{2222}x_2^3.$$
In the same way, we also have $g(x_1,x_2)\geq0$ for all $x\geq0$ by the assumption (2). So, $\mathcal{A}x^4=a_{1111}x_1^4+x_2g(x_1,x_2)\geq0\mbox{ for all }x=(x_1,x_2)^\top\geq0,$ and hence, the copositivity of $\mathcal{A}$ is proved, as required.
	\end{proof}	
Similarly, using Theorem \ref{th:32} (for the tensor $\Gamma=(\gamma_{ijk})$ in the above proof), the following theorem is obtained easily.
\begin{theorem} \label{th:42} Let $\mathcal{A}$ be an 4th order $2$ dimensional symmetric tensor with $a_{1111}\geq0$ and $ a_{2222}\geq0$. Assume that one of the following conditions (1) and (2) holds,
$$\begin{aligned}
(1)\ \ & a_{1222}\geq0,\ a_{1112}\geq \frac14a_{1111}-\sqrt{a_{1111} a_{1222}},\  a_{1221}\geq \frac23(a_{1222}-2\sqrt{a_{1111} a_{1222}});\\
(2)\ \ &a_{1112}\geq0,\ a_{1222}\geq \frac14a_{2222}-\sqrt{a_{1112} a_{2222}},\  a_{1221}\geq \frac23(a_{1112}-2\sqrt{a_{1112} a_{2222}}). \\
\end{aligned}$$
		Then $\mathcal{A}$ is copositive.
	\end{theorem}
Let $\mathcal{A}$ be an 4th order $3$ dimensional symmetric tensor. Then for a vector $x=(x_1,x_2,x_3)^\top$,
\begin{equation}\label{eq:42}\begin{aligned}\mathcal{A}x^4=&\sum_{i,j,k,l=1}^3a_{ijkl}x_ix_jx_kx_l\\
		=&a_{1111}x_1^4+a_{2222}x_2^4+a_{3333}x_3^4+4a_{1222}x_1x_2^3 +4a_{1333}x_1x_3^3\\
		& +4a_{2111}x_1^3x_2+4a_{2333}x_2x_3^3+4a_{3111}x_1^3x_3
		+4a_{3222}x_2^3x_3\\
		&+6a_{1122}x_1^2x_2^2+6a_{1133}x_1^2x_3^2+6a_{2233}x_2^2x_3^2\\
		&+12a_{1231}x_1^2x_2x_3+12a_{1232}x_1x_2^2x_3+12a_{1233}x_1x_2x_3^2.
		\end{aligned}\end{equation}

Now we give several sufficient conditions of  4th order copositive tensors with the help of ones of 3rd order copositive tensors in Section 3.	

\begin{theorem} \label{th:43} Let $\mathcal{A}$ be an 4th order $3$ dimensional symmetric tensor. Assume that
	$$\begin{aligned}
	& a_{1111}\geq 0,\ a_{2222}\geq 0, \ a_{3333}\geq 0, \\
&  a_{1112}\geq 0, \ a_{1113}\geq 0, \ a_{1222}\geq0, \ a_{2223}\geq 0, \ a_{1333}\geq 0,  \ a_{2333}\geq 0,\\
&\max\{a_{1222},a_{1333}\}>0, \max\{a_{1112},a_{2333}\}>0, \max\{a_{1113},a_{2223}\}>0,\\
& 6a_{1122}+\sqrt{a_{1111}a_{2222}}\geq0,\ 6a_{1133}+\sqrt{a_{1111}a_{3333}}\geq0,\ 6a_{2233}+\sqrt{a_{3333}a_{2222}}\geq0,\\
&8a_{1222}a_{1233}^3+8a_{1223}^3a_{1333}+16a_{1222}^2a_{1333}^2-24a_{1222}a_{1223}a_{1233}a_{1333}-3a_{1223}^2a_{1233}^2 \geq0,\\
&8a_{2111}a_{1233}^3+8a_{1123}^3a_{2333}+16a_{2111}^2a_{2333}^2-24a_{2111}a_{1123}a_{1233}a_{2333}-3a_{1123}^2a_{1233}^2 \geq0,\\
&8a_{3111}a_{1223}^3+8a_{1123}^3a_{2223}+16a_{3111}^2a_{3222}^2-24a_{3111}a_{1123}a_{1223}a_{3222}-3a_{1123}^2a_{1223}^2 \geq0.
	\end{aligned}$$
	Then $\mathcal{A}$ is  copositive.
\end{theorem}
\begin{proof}  Rewritten the homogeneous polynomial $\mathcal{A}x^4$ \eqref{eq:42} as follows,
\begin{align*}\mathcal{A}x^4=&(\frac12a_{1111}x_1^4+6a_{1122}x_1^2x_2^2+\frac12a_{2222}x_2^4)\\
+&(\frac12a_{2222}x_2^4+6a_{2233}x_2^2x_3^2+\frac12a_{3333}x_3^4)\\
+&(\frac12a_{1111}x_1^4+6a_{1133}x_1^2x_3^2+\frac12a_{3333}x_3^4)\\
+& x_1(4a_{1222}x_2^3+6a_{1223}x_2^2x_3 +6a_{1233}x_2x_3^2+4a_{1333}x_3^3)\\
+& x_2(4a_{2111}x_1^3+6a_{1123}x_1^2x_3+6a_{1233}x_1x_3^2+4a_{2333}x_3^3)\\
+& x_3(4a_{3111}x_1^3+6a_{1123}x_1^2x_2+6a_{1223}x_1x_2^2+4a_{3222}x_2^3).
		\end{align*}
Let $$p_1(x_1,x_2)=\frac12a_{1111}x_1^4+6a_{1122}x_1^2x_2^2+\frac12a_{2222}x_2^4.$$
It follows from Lemma \ref{le:21} that we can restrict $x$ to
		$$\|x\|=x_1+x_2+x_3=1\mbox{ and }x_i\geq0\mbox{ for }i=1,2,3.$$
Clearly, $p_1(0,0)=0,\ p_1(x_1,0)=\frac12a_{1111}x_1^4\geq0, \ p_1(0,x_2)=\frac12a_{2222}x_2^4\geq0.$ Let we assume that both $x_1$ and $x_2$ are not zero.
Without loss of generality, let $x_2>0$ and $t=\frac{x_1^2}{x_2^2}$. Suppose
$$p(t)=\frac{p_1(x_1,x_2)}{x_2^4}=\alpha t^2+\beta t+\gamma\mbox{ with }\alpha=\frac12a_{1111},\ \beta=6a_{1122},\ \gamma=\frac12a_{2222}.$$
It follows from Lemma \ref{le:23} that the assumptions $a_{1111}\geq0$ and $a_{2222}\geq0$ together with the inequality $6a_{1122}+\sqrt{a_{1111}a_{2222}}\geq0$ mean $p(t)\geq0$ for all $t\geq0$, and hence, $p_1(x_1,x_2)\geq0$ for all $x=(x_1,x_2)^\top\geq0$.

Similarly, we also have $$p_2(x_2,x_3)=\frac12a_{2222}x_2^4+6a_{2233}x_2^2x_3^2+\frac12a_{3333}x_3^4\geq0$$ and $$p_3(x_1,x_3)=\frac12a_{1111}x_1^4+6a_{1133}x_1^2x_3^2+\frac12a_{3333}x_3^4\geq0.$$

Let $$F_1(x_2,x_3)=4a_{1222}x_2^3+6a_{1223}x_2^2x_3 +6a_{1233}x_2x_3^2+4a_{1333}x_3^3.$$ Then the homogeneous polynomial $F_1(x_2,x_3)$ may be written as
$$F_1(x_2,x_3)=\Gamma y^3=\sum_{i,j,k=1}^{2}\gamma_{ijk}y_iy_jy_k\mbox{ for }y=(x_2,x_3)^\top,$$
where $\Gamma=(\gamma_{ijk})$ is 3rd order 2 dimensional symmetric tensor with its entries
$$\gamma_{111}=4a_{1222},\ \gamma_{112}=2a_{1223},\ \gamma_{122}=2a_{1233},\ \gamma_{222}=4a_{1333}.$$
Then by assumptions, we have \begin{align*}&4\gamma_{111}\gamma_{122}^3+4\gamma_{112}^3\gamma_{222}+\gamma_{111}^2\gamma_{222}^2-6\gamma_{111}\gamma_{112}\gamma_{122}\gamma_{222}-3\gamma_{112}^2\gamma_{122}^2 \\=& 16( 8a_{1222}a_{1233}^3+8a_{1223}^3a_{1333}+16a_{1222}^2a_{1333}^2-24a_{1222}a_{1223}a_{1233}a_{1333}\\
&\ -3a_{1223}^2a_{1233}^2)\geq0.\end{align*}
From Theorem \ref{th:31}, it follows that $\Gamma$ is copositive, i.e., $$F_1(x_2,x_3)=\Gamma y^3\geq0\mbox{ for all }y=(x_2,x_3)^\top\geq0.$$

Similarly, we must have $$F_2(x_1,x_3)=4a_{2111}x_1^3+6a_{1123}x_1^2x_3+6a_{1233}x_1x_3^2+4a_{2333}x_3^3\geq0$$ and
$$F_3(x_1,x_2)=4a_{3111}x_1^3+6a_{1123}x_1^2x_2+6a_{1223}x_1x_2^2+4a_{3222}x_2^3\geq0.$$
Thus, for all $x=(x_1,x_2,x_3)^\top\geq0,$ we have \begin{align*}\mathcal{A}x^4=&p_1(x_1,x_2)+p_2(x_2,x_3)+p_3(x_1,x_3)\\ +&x_1F_1(x_2,x_3)+x_2F_2(x_1,x_3)+x_3F_3(x_1,x_2)\geq0.\end{align*}
Namely, $\mathcal{A}$ is copositive, as required.
\end{proof}	

Now we give a simpler sufficient condition of copositive tensors using Theorem \ref{th:32} (for the tensor $\Gamma=(\gamma_{ijk})$ in the above proof).

\begin{theorem} \label{th:44} Let $\mathcal{A}$ be an 4th order $3$ dimensional symmetric tensor. Assume that
	$$\begin{aligned}
	& a_{1111}\geq 0,\ a_{2222}\geq 0, \ a_{3333}\geq 0,\\
&  a_{1112}\geq 0, \ a_{1113}\geq 0, \ a_{1222}\geq0, \ a_{2223}\geq 0, \ a_{1333}\geq 0,  \ a_{2333}\geq 0,\\
& a_{1122}\geq-\frac16\sqrt{a_{1111}a_{2222}},\ a_{1133}\geq-\frac16\sqrt{a_{1111}a_{3333}},\ a_{2233}\geq-\frac16\sqrt{a_{3333}a_{2222}},\\
&a_{1223}\geq\frac23\max\{a_{1222}-2\sqrt{a_{1222}a_{1333}},a_{2223}-2\sqrt{a_{1112}a_{2223}}\},\\
& a_{1233}\geq\frac23\max\{a_{2333}-2\sqrt{a_{1112}a_{2333}},a_{1333}-2\sqrt{a_{1222}a_{1333}}\},\\
&a_{1123}\geq\frac23\max\{a_{1113}-2\sqrt{a_{1113}a_{2223}},a_{1112}-2\sqrt{a_{1112}a_{2333}}\}.
	\end{aligned}$$
	Then $\mathcal{A}$ is  copositive.
\end{theorem}
	\begin{theorem} \label{th:45} Let $\mathcal{A}$ be an 4th order $3$ dimensional symmetric tensor. Assume that
	$$\begin{aligned}
	& a_{1111}>0,\ a_{2222}> 0, \ a_{3333}> 0,\ a_{1113}\geq 0, \ a_{1222}\geq0,  \ a_{2333}\geq 0,\\
&9a_{1122}+\sqrt{a_{1111}a_{2222}}\geq 0, \ 9a_{1133}+\sqrt{a_{1111}a_{3333}}\geq 0, \ 9a_{2233}+\sqrt{a_{3333}a_{2222}}\geq 0,\\
&27a_{1123}+\sqrt{\left(9a_{1122}+\sqrt{a_{1111}a_{2222}}\right)\left(9a_{1133}+\sqrt{a_{1111}a_{3333}}\right)}\geq0,\\
&27a_{1223}+\sqrt{\left(9a_{1122}+\sqrt{a_{1111}a_{2222}}\right)\left(9a_{2233}+\sqrt{a_{3333}a_{2222}}\right)}\geq0,\\
&27a_{1233}+\sqrt{\left(9a_{1133}+\sqrt{a_{1111}a_{3333}}\right)\left(9a_{2233}+\sqrt{a_{3333}a_{2222}}\right)}\geq0,\\
&2a_{1111}\left(9a_{1122}+\sqrt{a_{1111}a_{2222}}\right)^3+3^{7}\cdot4^{3}a_{1222}a_{1112}^3+3^{8}a_{1111}^2a_{1222}^2\\
&-3^{6}\cdot4a_{1111}a_{1112}a_{1222}\left(9a_{1122}+\sqrt{a_{1111}a_{2222}}\right)-3^{3}\cdot4a_{1112}^2\left(9a_{1122}+\sqrt{a_{1111}a_{2222}}\right)^2\geq0\\
&2a_{2222}\left(9a_{2233}+\sqrt{a_{3333}a_{2222}}\right)^3+3^{7}\cdot4^{3}a_{2223}a_{2333}^3+3^{8}a_{2222}^2a_{2333}^2\\
&-3^{6}\cdot4a_{2222}a_{2223}a_{2333}\left(9a_{2233}+\sqrt{a_{3333}a_{2222}}\right)-3^{3}\cdot4a_{2223}^2\left(9a_{2233}+\sqrt{a_{3333}a_{2222}}\right)^2\geq0\\
&2a_{3333}\left(9a_{1133}+\sqrt{a_{222}a_{3333}}\right)^3+3^{7}\cdot4^{3}a_{1333}a_{1113}^3+3^{8}a_{3333}^2a_{1333}^2\\
&-3^{6}\cdot4a_{3333}a_{1113}a_{1333}\left(9a_{1133}+\sqrt{a_{1111}a_{3333}}\right)-3^{3}\cdot4a_{1333}^2\left(9a_{1133}+\sqrt{a_{1111}a_{3333}}\right)^2\geq0.
	\end{aligned}$$
	Then $\mathcal{A}$ is  copositive.
\end{theorem}
\begin{proof}  Rewritten the homogeneous polynomial $\mathcal{A}x^4$ \eqref{eq:42} as follows,
\begin{align*}\mathcal{A}x^4=&\frac13\left(\left(\sqrt{a_{1111}}x_1^2-\sqrt{a_{2222}}x_2^2\right)^2+\left(\sqrt{a_{1111}}x_1^2-\sqrt{a_{3333}}x_3^2\right)^2
+\left(\sqrt{a_{3333}}x_3^2-\sqrt{a_{2222}}x_2^2\right)^2\right)\\
+&\frac19x_1\left(3a_{1111}x_1^3+36a_{1112}x_1^2x_2+2\left(9a_{1122}+\sqrt{a_{1111}a_{2222}}\right)x_1x_2^2+36a_{1222}x_2^3\right)\\
+&\frac19x_2\left(3a_{2222}x_2^3+36a_{2223}x_2^2x_3+2\left(9a_{2233}+\sqrt{a_{3333}a_{2222}}\right)x_2x_3^2+36a_{2333}x_3^3\right)\\
+&\frac19x_3\left(3a_{3333}x_3^3+36a_{1333}x_1x_3^2+2\left(9a_{1133}+\sqrt{a_{1111}a_{3333}}\right)x_1^2x_3+36a_{1113}x_1^3\right)\\
+& \frac29x_1^2\left(\left(9a_{1122}+\sqrt{a_{1111}a_{2222}}\right)x_2^2+54a_{1123}x_2x_3+\left(9a_{1133}+\sqrt{a_{1111}a_{3333}}\right)x_3^2\right)\\
+& \frac29x_2^2\left(\left(9a_{1122}+\sqrt{a_{1111}a_{2222}}\right)x_1^2+54a_{1223}x_1x_3+\left(9a_{2233}+\sqrt{a_{3333}a_{2222}}\right)x_3^2\right)\\
+& \frac29x_3^2\left(\left(9a_{1133}+\sqrt{a_{1111}a_{3333}}\right)x_1^2+54a_{1233}x_1x_2+\left(9a_{2233}+\sqrt{a_{3333}a_{2222}}\right)x_2^2\right).
		\end{align*}
Let $$F_1(x_1,x_2)=3a_{1111}x_1^3+36a_{1112}x_1^2x_2+2\left(9a_{1122}+\sqrt{a_{1111}a_{2222}}\right)x_1x_2^2+36a_{1222}x_2^3.$$ Then the homogeneous polynomial $F_1(x_1,x_2)$ may be written as
$$F_1(x_1,x_2)=\Gamma y^3=\sum_{i,j,k=1}^{2}\gamma_{ijk}y_iy_jy_k\mbox{ for }y=(x_1,x_2)^\top,$$
where $\Gamma=(\gamma_{ijk})$ is 3rd order 2 dimensional symmetric tensor with its entries
$$\gamma_{111}=3a_{1111},\ \gamma_{112}=12a_{1112},\ \gamma_{122}=\frac23\left(9a_{1122}+\sqrt{a_{1111}a_{2222}}\right),\ \gamma_{222}=36a_{1222}.$$
Then by assumptions, we have
\begin{align*}&4\gamma_{111}\gamma_{122}^3+4\gamma_{112}^3\gamma_{222}+\gamma_{111}^2\gamma_{222}^2-6\gamma_{111}\gamma_{112}\gamma_{122}\gamma_{222}-3\gamma_{112}^2\gamma_{122}^2 \\
=& 4(3a_{1111})\left(\frac23\left(9a_{1122}+\sqrt{a_{1111}a_{2222}}\right)\right)^3+4(12a_{1112})^3(36a_{1222})+(3a_{1111})^2(36a_{1222})^2\\
&-6(3a_{1111})(12a_{1112})(\frac23\left(9a_{1122}+\sqrt{a_{1111}a_{2222}}\right))(36a_{1222})\\
&-3(12a_{1112})^2\left(\frac23\left(9a_{1122}+\sqrt{a_{1111}a_{2222}}\right)\right)^2\\
=&\frac{16}{9}(2a_{1111}\left(9a_{1122}+\sqrt{a_{1111}a_{2222}}\right)^3+3^{7}\cdot4^{3}a_{1222}a_{1112}^3+3^{8}a_{1111}^2a_{1222}^2\\
-&3^{6}\cdot4a_{1111}a_{1112}a_{1222}\left(9a_{1122}+\sqrt{a_{1111}a_{2222}}\right)-3^{3}\cdot4a_{1112}^2\left(9a_{1122}+\sqrt{a_{1111}a_{2222}}\right)^2)\geq0.\end{align*}
From Theorem \ref{th:31}, it follows that $\Gamma$ is copositive, i.e., $$F_1(x_1,x_2)=\Gamma y^3\geq0\mbox{ for all }y=(x_1,x_2)^\top\geq0.$$

Similarly, we must have $$F_2(x_2,x_3)=3a_{2222}x_2^3+36a_{2223}x_2^2x_3+2\left(9a_{2233}+\sqrt{a_{3333}a_{2222}}\right)x_2x_3^2+36a_{2333}x_3^3\geq0$$ and
$$F_3(x_1,x_3)=3a_{3333}x_3^3+36a_{1333}x_1x_3^2+2\left(9a_{1133}+\sqrt{a_{1111}a_{3333}}\right)x_1^2x_3+36a_{1113}x_1^3\geq0.$$
Let $$p_1(x_2,x_3)=\left(9a_{1122}+\sqrt{a_{1111}a_{2222}}\right)x_2^2+54a_{1123}x_2x_3+\left(9a_{1133}+\sqrt{a_{1111}a_{3333}}\right)x_3^2.$$
It follows from (strict) copositivity of $2\times2$ matrix (or Lemma \ref{le:23}) that $p_1(x_2,x_3)\geq0$ for all $(x_2,x_3)^\top\geq0$ if and only if
		$$\begin{aligned}&9a_{1122}+\sqrt{a_{1111}a_{2222}}\geq0,\ 9a_{1133}+\sqrt{a_{1111}a_{3333}}\geq0,\\
& 27a_{1123}+\sqrt{\left(9a_{1122}+\sqrt{a_{1111}a_{2222}}\right)\left(9a_{1133}+\sqrt{a_{1111}a_{3333}}\right)}\geq0.\end{aligned}$$
Similarly, we also have $$p_2(x_1,x_3)=\left(9a_{1122}+\sqrt{a_{1111}a_{2222}}\right)x_1^2+54a_{1223}x_1x_3+\left(9a_{2233}+\sqrt{a_{3333}a_{2222}}\right)x_3^2\geq0$$
and $$p_3(x_1,x_2)=\left(9a_{1133}+\sqrt{a_{1111}a_{3333}}\right)x_1^2+54a_{1233}x_1x_2+\left(9a_{2233}+\sqrt{a_{3333}a_{2222}}\right)x_2^2\geq0.$$

Thus, for all $x=(x_1,x_2,x_3)^\top\geq0,$ we have
\begin{align*}\mathcal{A}x^4=&\frac13\left(\sqrt{a_{1111}}x_1^2-\sqrt{a_{2222}}x_2^2\right)^2+\frac13\left(\sqrt{a_{1111}}x_1^2-\sqrt{a_{3333}}x_3^2\right)^2\\&
+\frac13\left(\sqrt{a_{3333}}x_3^2-\sqrt{a_{2222}}x_2^2\right)^2\\&+\frac19(x_1F_1(x_1,x_2)+x_2F_2(x_2,x_3)+x_3F_3(x_1,x_3))\\
&+\frac29(x_1^2p_1(x_2,x_3)+2x_2^2p_2(x_1,x_3)+2x_3^2p_3(x_1,x_2))\geq0.\end{align*}
Namely, $\mathcal{A}$ is copositive, as required.
\end{proof}

\begin{remark} In this section, the conclusions are proved by reducing dimensions or orders of tensor. For example, Theorems \ref{th:45} and \ref{th:43}, an 4th order 3 dimensional tensor is decomposed into three 3rd order 2 dimensional tensors and three 2nd order 2 dimensional tensors, and then, by studying the copositivity of these lower dimensional and lower order tensors,  the desired conclusions are proved. Similarly, we may decompose an 4th order 3 dimensional tensor into three 3rd order 3 dimensional tensors $\mathbf{\Gamma}_1$, $\mathbf{\Gamma}_2$, $\mathbf{\Gamma}_3$,
\begin{align*}\mathcal{A}x^4=&x_1(a_{1111}x_1^3+2a_{1112}x_1^2x_2+3a_{1122}x_1x_2^2+2a_{1222}x_2^3+2a_{1333}x_3^3\\
&+3a_{1133}x_1x_3^2+2a_{1113}x_1^2x_3+4a_{1123}x_1x_2x_3+4a_{1223}x_2^2x_3+4a_{1233}x_2x_3^2)\\
+&x_2(2a_{1112}x_1^3+3a_{1122}x_1^2x_2+2a_{1222}x_1x_2^2+a_{2222}x_2^3+2a_{2333}x_3^3\\
&+4a_{1233}x_1x_3^2+4a_{1123}x_1^2x_3+4a_{1223}x_1x_2x_3+2a_{2223}x_2^2x_3+3a_{2233}x_2x_3^2)\\
+&x_3(2a_{1113}x_1^3+4a_{1123}x_1^2x_2+4a_{1223}x_1x_2^2+2a_{2223}x_2^3+a_{3333}x_3^3\\
&+2a_{1333}x_1x_3^2+3a_{1133}x_1^2x_3+4a_{1233}x_1x_2x_3+3a_{2233}x_2^2x_3+2a_{2333}x_2x_3^2)\\
=&x_1\mathbf{\Gamma}_1x^3+x_2\mathbf{\Gamma}_2x^3+x_3\mathbf{\Gamma}_3x^3.
\end{align*}
Then we can obtain other sufficient conditions of copositivity of an 4th order 3 dimensional tensor $\mathcal{A}$ by Theorems \ref{th:34} and \ref{th:35}.  Here we omit it.
This way to reducing dimensions or orders of tensor may be a very important method of analysing higher order tensors in future.
\end{remark}
	
\section{\bf Checking vacuum stability of $\mathbb{Z}_3$ scalar dark matter}
Kannike \cite{K2016,K2018} presented a physical example defined by scalar dark matter stable under a $\mathbb{Z}_3$ discrete group, that is,
the most general scalar quartic
potential of the {\bf SM} Higgs $\mathbf{H}_1$, an inert doublet $\mathbf{H}_2$ and a
complex singlet $\mathbf{S}$ which is symmetric under a $\mathbb{Z}_3$ group is
\begin{align}
V(h_1,h_2,s)=& \lambda_1|H_1|^4+\lambda_2|H_2|^4+\lambda_3|H_1|^2|H_2|^2+\lambda_4(H_1^{\dagger} H_2)(H_2^{\dagger} H_1)\nonumber\\
&\ +\lambda_S|S|^4+\lambda_{S1}|S|^2|H_1|^2+\lambda_{S2}|S|^2|H_2|^2\nonumber\\
&\ +\frac12(\lambda_{S12}S^2H_1^{\dagger} H_2+\lambda_{S12}^*S^{{\dagger}2}H_2^{\dagger} H_1)\nonumber\\
= & \lambda_1h_1^4+\lambda_2h_2^4+\lambda_3h_1^2h^2_2+\lambda_4\rho^2 h_1^2h_2^2\nonumber\\
&\ +\lambda_Ss^4+\lambda_{S1}s^2h_1^2+\lambda_{S2}s^2h_2^2-|\lambda_{S12}|\rho s^2h_1h_2\nonumber\\
=&\mathbf{\Gamma} z^4=\sum_{i,j,k,l=1}^{3}\gamma_{ijkl}z_kz_jz_kz_l,\nonumber
\end{align}
where $z=(z_1,z_2,z_3)^\top=(h_1,h_2,s)^\top$, the orbit space parameter $\rho\in[0,1]$, $$
h_1=|H_1|,\ h_2=|H_2|, H_2^{\dagger} H_1=h_1h_2\rho e^{i\phi}, S=se^{i\phi_S}, \lambda_{S12}=-|\lambda_{S12}|,$$
$\mathbf{\Gamma}=(\gamma_{ijkl})$ is an 4th order 3 dimensional real symmetric tensor with
\begin{align}
&\gamma_{1111}=\lambda_1,\ \gamma_{2222}=\lambda_2,\ \gamma_{3333}=\lambda_S, \nonumber\\
&\gamma_{1122}=\frac16(\lambda_3+\lambda_4\rho^2),\ \gamma_{1133}=\frac16\lambda_{S1},\ \gamma_{2233}=\frac16\lambda_{S2}, \nonumber\\
& \gamma_{1233}=-\frac1{12}|\lambda_{S12}|\rho,\ \gamma_{ijkl}=0 \mbox{ for the others}.\nonumber
\end{align}
Clearly, $z\geq0.$ It follows from Theorem \ref{th:45} that the conditions of (strict) copositivity of the tensor $\mathbf{\Gamma}$ (that is, $V(h_1,h_2,s)= \mathbf{\Gamma}z^4\geq0(>0)$) are
\begin{align}
&\lambda_1>0,\ \lambda_2>0,\ \lambda_S>0, \nonumber\\
&3\lambda_3+3\lambda_4\rho^2+2\sqrt{\lambda_1\lambda_2}\geq0(>0),\ 3\lambda_{S1}+2\sqrt{\lambda_1\lambda_S}\geq0(>0),\nonumber\\ & 3\lambda_{S2}+2\sqrt{\lambda_S\lambda_2}\geq0(>0), \nonumber\\
& -\frac94|\lambda_{S12}|\rho+\sqrt{\left(3\lambda_{S1}+2\sqrt{\lambda_1\lambda_S}\right)\left(3\lambda_{S2}+2\sqrt{\lambda_S\lambda_2}\right)}\geq0(>0).\nonumber
\end{align}
So these conditions assure the potential $V(h_1,h_2,s)$ under a $\mathbb{Z}_3$ group is bounded from below, and hence, this guarantees the vacuum stability for $\mathbb{Z}_3$ scalar dark matter. 

\bibliographystyle{amsplain}

\end{document}